\documentclass[11pt]{article}
\usepackage{authblk}
  \usepackage{fullpage}
  \usepackage{amsmath,amsfonts,amsthm}
  \usepackage{tikz}
  \usepackage[T1]{fontenc}
  \usepackage[utf8]{inputenc}
  \usepackage{comment}
  \usepackage{xspace}
  \usepackage{enumerate}
  \usepackage{listings}
  \usepackage[hidelinks]{hyperref}
  \usepackage[linesnumbered, noend]{algorithm2e}
  \newcommand{\floor}[1]{\left\lfloor #1 \right\rfloor}
  
  \newcommand{\Oh}{\mathcal{O}}

  \newcommand{\pmilp}{$(\pm)\mathrm{ILP}$\xspace}

	\newcommand{\B}{\mathcal{B}}
	\newcommand{\M}{\mathcal{M}}
	\newcommand{\T}{\mathcal{T}}

  \renewcommand{\a}{\mathbf{a}}
  \newcommand{\x}{\mathbf{x}}
  \newcommand{\y}{\mathbf{y}}
  \newcommand{\s}{\mathbf{s}}
  \newcommand{\A}{\mathcal{A}}
  \newcommand{\E}{{E}}
  \newcommand{\F}{\mathcal{F}}
  \newcommand{\G}{\mathcal{G}}
  \newcommand{\dist}{\mathit{dist}}
  
  \newcommand{\sub}{\subseteq}
  \newcommand{\pint}{\mathbb{Z}_{\ge 0}}

  \newcommand{\I}{\mathcal{I}}
  \newcommand{\OPT}{\mathrm{OPT}}
  \newcommand{\ILP}{\mathrm{ILP}}

  \newtheorem{definition}{Definition}
  \newtheorem{observation}{Observation}
  \newtheorem{fact}{Fact}
  \newtheorem{example}{Example}
  \newtheorem{corollary}{Corollary}
  \newtheorem{lemma}{Lemma}
  \newtheorem{theorem}{Theorem}
  \newtheorem{remark}{Remark}
  \newtheorem{claim}{Claim}

  \date{}

  \title{On the String Consensus Problem and the Manhattan Sequence Consensus Problem}
 
\author[1]{Tomasz Kociumaka}
\author[2]{Jakub W. Pachocki}
\author[1]{Jakub Radoszewski}
\author[1,3]{Wojciech Rytter}
\author[1]{Tomasz Wale\'n}
 
 \affil[1]{Faculty of Mathematics, Informatics and Mechanics, University of Warsaw, Poland \texttt{[kociumaka, jrad, rytter, walen]@mimuw.edu.pl}}
 \affil[2]{Carnegie Mellon University, \texttt{pachocki@cs.cmu.edu}}
 \affil[3]{Faculty of Mathematics and Computer Science, Copernicus University, Toru\'n, Poland}

\begin{document}
  \maketitle
\begin{abstract}
In the \textsc{Manhattan Sequence Consensus}  problem (MSC problem) we are given  $k$ integer sequences,
each of length $\ell$, and we are to find an integer sequence ${\x}$ of length $\ell$
(called a consensus sequence), such that the maximum
Manhattan distance of ${\x}$
from each of the input  sequences  is minimized.
For binary sequences Manhattan distance coincides with Hamming distance, hence in this case
the string consensus problem (also called string center problem or closest string problem) is a special
case of MSC.
Our main result is a practically efficient $\Oh(\ell)$-time algorithm solving
MSC for $k\le 5$ sequences. Practicality of our algorithms has been
verified experimentally.
It improves upon the quadratic algorithm by Amir et al.\ (SPIRE 2012)
for string consensus problem for $k=5$ binary strings.
Similarly as in Amir's algorithm we use a column-based framework.
We replace the implied general integer linear programming
by its easy special cases,
due to combinatorial properties of the MSC for $k\le 5$.
We also show that for a general parameter $k$
any instance can be reduced in linear time to
a kernel of size $k!$, so the problem is fixed-parameter tractable.
Nevertheless, for $k\ge 4$ this is still too large for
any naive solution to be feasible in practice.
\end{abstract}

  \section{Introduction}
  In the sequence consensus problems, given a set of sequences of length $\ell$ we are searching for
  a new sequence of length $\ell$ which minimizes the maximum distance to all the
  given sequences in some particular metric.
  Finding the consensus sequence is a tool for many clustering algorithms
  and as such has applications in unsupervised learning, classification, databases,
  spatial range searching, data mining etc \cite{DBLP:conf/stoc/BadoiuHI02}.
  It is also one of popular methods for detecting data commonalities of many strings
  (see \cite{DBLP:conf/spire/AmirPR12}) and has a considerable number of applications
  in coding theory \cite{DBLP:journals/tit/CohenHLS97,DBLP:journals/mst/FrancesL97},
  data compression \cite{DBLP:journals/tit/GrahamS85}
  and bioinformatics \cite{DBLP:journals/algorithmica/GrammNR03,DBLP:conf/soda/LanctotLMWZ99}.
  The consensus problem has previously been studied mainly in $\mathbb{R}^\ell$ space with the Euclidean distance
  and in $\Sigma^\ell$ (that is, the space of sequences over a finite alphabet $\Sigma$) with the Hamming distance.
  Other metrics were considered in \cite{DBLP:journals/ipl/AmirPR13}.
  We study the sequence consensus problem for Manhattan metric ($\ell_1$ norm) in correlation with the Hamming-metric
  variant of the problem.

  The Euclidean variant of the sequence consensus problem is also known as the bounding sphere, enclosing sphere or enclosing ball problem.
  It was initially introduced in 2 dimensions (i.e., the smallest circle problem) by Sylvester in 1857 \cite{sylvester}.
  For an arbitrary number of dimensions, a number of approximation
  algorithms \cite{DBLP:conf/stoc/BadoiuHI02,Kumar:2003:AME:996546.996548,Ritter} 
  and practical exact algorithms \cite{DBLP:conf/esa/FischerGK03,DBLP:conf/compgeom/GartnerS00} for this problem
  have been proposed.

  The Hamming distance variant of the sequence consensus problem is known under the names of string consensus,
  center string or closest string problem.
  The problem is known to be NP-complete even for binary alphabet \cite{DBLP:journals/mst/FrancesL97}.
  The algorithmic study of Hamming string consensus (HSC) problem started in 1999 with the first
  approximation algorithms \cite{DBLP:conf/soda/LanctotLMWZ99}.
  Afterwards polynomial-time approximation schemes (PTAS) with different running times
  were presented in \cite{DBLP:conf/focs/AndoniIP06,DBLP:journals/siamcomp/MaS09,DBLP:conf/isit/MazumdarPS13}.
  A number of exact algorithms have also been presented.
  Many of these consider a decision version of the problem, in which we are to check if there is
  a solution to HSC problem with distance at most $d$ to the input sequences.
  Thus FPT algorithms with time complexities $\Oh(k\ell+kd^{d+1})$ and $\Oh(k\ell+kd(16|\Sigma|)^d)$ were presented
  in \cite{DBLP:journals/algorithmica/GrammNR03} and \cite{DBLP:journals/siamcomp/MaS09}, respectively.

  An FPT algorithm parameterized only by $k$ was given in \cite{DBLP:journals/algorithmica/GrammNR03}.
  It uses Lenstra's algorithm \cite{Lenstra} for a solution of an integer linear program of size exponential in $k$
  (which requires $\Oh(k!^{4.5k!}\ell)$ operations on integers of magnitude $\Oh(k!^{2k!}\ell)$, see \cite{DBLP:conf/spire/AmirPR12})
  and due to extremely large constants is not feasible for $k \ge 4$.
  This opened a line of research with efficient algorithms for small constant $k$.
  A linear-time algorithm for $k=3$ was presented in \cite{DBLP:journals/algorithmica/GrammNR03},
  a linear-time algorithm for $k=4$ and binary alphabet was given in \cite{DBLP:conf/spire/BoucherBD08},
  and recently an $\Oh(\ell^2)$-time algorithm for $k=5$ and also binary alphabet was developed in \cite{DBLP:conf/spire/AmirPR12}.

  For two sequences $\x=(x_1,\ldots,x_\ell)$ and $\y=(y_1,\ldots,y_\ell)$ the Manhattan distance (also known as rectilinear
  or taxicab distance) between
  $\x$ and $\y$ is defined as follows:
  $$\dist(\x,\y) = \sum_{j=1}^\ell |x_j-y_j|.$$
  The Manhattan version of the consensus problem is formally defined as follows:

  \vskip 0.1cm \noindent {{\scshape{Manhattan Sequence Consensus}} \bf problem}
  \begin{description}

\item[Input:] A collection $\A$ of $k$ integer sequences $\a_i$, each of length $\ell$;

\vskip 0.1cm

\item[Output:] $\OPT(\A)\;=\; \min_{\x}\;\max\; \{\dist(\x,\a_i)\;:\; 1\le i\le k\}$,\\ and the corresponding integer consensus sequence $\x$.
  \end{description}
We assume that integers $a_{i,j}$ satisfy $|a_{i,j}|\le M$, and all $\ell,k,M$ fit in a machine word,
so that arithmetics on integers of magnitude $\Oh(\ell M)$ take constant time.

  For simplicity in this version of the paper we concentrate on computing $\OPT(\A)$
  and omit the details of recovering the corresponding consensus sequence $\x$.
  Nevertheless, this step is included in the implementation provided.

  \begin{example}
    Let $\A=((120,0,80),\;(20,40,130),\;(0,100,0))$.
    Then $\OPT(\A)=150$ and a consensus sequence is $\x = (30,40,60)$, see also Fig.~\ref{fig:first_example}.
  \end{example}

  \noindent
  Our results are the following:
  \begin{itemize}
    \item
      We show that {\scshape{Manhattan Sequence Consensus}} problem has a kernel with $\ell\le k!$
      and is fixed-parameter linear with the parameter $k$.
    \item
      We present a practical linear-time algorithm for the {\scshape{Manhattan Sequence Consensus}} problem
      for $k = 5$ (which obviously can be used for any $k \le 5$).
  \end{itemize}

  \noindent
  Note that binary HSC problem is a special case of MSC problem.
  Hence, the latter problem is NP-complete.
  Moreover, the efficient linear-time algorithm presented here for MSC problem for $k=5$ yields
  an equally efficient linear-time algorithm for the binary HSC problem and thus improves the result of \cite{DBLP:conf/spire/AmirPR12}.

  \paragraph{Organization of the paper.}
  Our approach is based on a reduction of the MSC problem to instances of integer linear programming (ILP).
  For general constant $k$ we obtain a constant, though a very large, number of instances
  with a constant number of variables that we solve using Lenstra's algorithm \cite{Lenstra} which works in constant time
  (the constant coefficient of this algorithm is also very large).
  This idea is similar to the one used in the FPT algorithm for
  HSC problem \cite{DBLP:journals/algorithmica/GrammNR03},
  however for MSC it requires an additional combinatorial observation.
  For $k \le 5$ we obtain a more efficient reduction of MSC to at most
  20 instances of very special ILP which we solve efficiently without applying a general ILP solver.

  In Section~\ref{sec:prelim} we show the first steps of the reduction of MSC to ILP.
  In Section~\ref{sec:ker} we show a kernel for the problem of $\Oh(k!)$ size.
  In Section~\ref{sec:comb5} we perform a combinatorial analysis of the case $k=5$
  which leaves 20 simple types of the sequence $\x$ to be considered.
  This analysis is used in Section~\ref{sec:algo5} to obtain 20 special ILP instances with only 4 variables.
  They could be solved using Lenstra's ILP solver.
  However, there exists an efficient algorithm tailored for this type of special instances.
  It can be found in Section~\ref{app:ilp}.
  Finally we analyze the performance of a C++ implementation of our algorithm
  in the Conclusions section.

  \section{From MSC Problem to ILP}\label{sec:prelim}
  Let us fix a collection $\A=(\a_1,\ldots,\a_k)$ of the input sequences.
  The elements of $\a_i$ are denoted by $a_{i,j}$ (for $1\le j \le \ell$).
  We also denote $\dist(\x,\A) = \max\; \{\dist(\x,\a_i)\;:\; 1\le i\le k\}$.

  For $j\in \{1,\ldots,n\}$
  let $\pi_j$ be a permutation of $\{1,\ldots,k\}$ such that
  $a_{\pi_j(1),j}\le \ldots \le a_{\pi_j(k),j}$,
  i.e. $\pi_j$ is the ordering permutation of elements $a_{1,j},\ldots,a_{k,j}$.
  We also set $s_{i,j} = a_{\pi_j(i),j}$, see Example~\ref{ex:first_example}.
  For some $j$ there might be several possibilities for $\pi_j$ (if $a_{i,j}=a_{i',j}$ for some $i\ne i'$),
  we fix a single choice for each~$j$.

  \begin{example}\label{ex:first_example}
    Consider the following 3 sequences $\a_i$ and sequences $\s_i$ obtained by sorting columns:

    $$[a_{i,j}] =
      \begin{bmatrix}
        120 & 0 & 80\\
        20 & 40 & 130\\
        0 & 100 & 0
      \end{bmatrix},
      \quad
      [s_{i,j}] =
      \begin{bmatrix}
        0 & 0 & 0\\
        20 & 40 & 80\\
        120 & 100 & 130
      \end{bmatrix}.
      $$

    The Manhattan consensus sequence is $\x=(30,40,60)$, see Fig.~\ref{fig:first_example}.
    In the figure, the circled numbers in $j$-th column are
    $\pi_j(1),\pi_j(2),\ldots, \pi_j(k)$  (top-down).
    \end{example}

    \begin{figure}[htb]
    \begin{center}
      \begin{tikzpicture}[scale=0.4,yscale=-1,xscale=0.5]
  \draw[rounded corners, thick, densely dotted] (0,2.025) -- (7,2.7) -- (14,4.05);
  \filldraw[white] (7,3) circle (1cm);
  \draw[-latex,thick,yshift=1cm] (-5,0) -- (-5,6);
  \draw (0,0) node[circle,draw,inner sep=2] (A31) {\scriptsize 3};
  \draw (0,1.35) node[circle,draw,inner sep=2] (A21) {\scriptsize 2};
  \draw (0,8.1) node[circle,draw,inner sep=2] (A11) {\scriptsize 1};
  \draw (7,0) node[circle,draw,inner sep=2] (A12) {\scriptsize 1};
  \draw (7,2.7) node[circle,draw,inner sep=2] (A22) {\scriptsize 2};
  \draw (7,6.75) node[circle,draw,inner sep=2] (A32) {\scriptsize 3};
  \draw (14,0) node[circle,draw,inner sep=2] (A33) {\scriptsize 3};
  \draw (14,5.4) node[circle,draw,inner sep=2] (A13) {\scriptsize 1};
  \draw (14,8.775) node[circle,draw,inner sep=2] (A23) {\scriptsize 2};
  \draw (0,0) node[left=0.2cm] {0};
  \draw (0,1.35) node[left=0.2cm] {20};
  \draw (0,8.1) node[left=0.2cm] {120};
  \draw (7,0) node[left=0.2cm] {0};
  \draw (7,2.9) node[left=0.2cm] {40};
  \draw (7,6.75) node[left=0.2cm] {100};
  \draw (14,0) node[right=0.2cm] {0};
  \draw (14,5.4) node[right=0.2cm] {80};
  \draw (14,8.775) node[right=0.2cm] {130};
  \draw (A31) -- (A21)  (A21) -- (A11);
  \draw (A12) -- (A22)  (A22) -- (A32);
  \draw (A33) -- (A13)  (A13) -- (A23);
  \filldraw (0,2.025) circle (0.15cm);
  \filldraw (14,4.05) circle (0.15cm);
  \draw (0,2.025) node[below left] {$x_1$};
  \draw (7,2.7) node[right=0.2cm] {$x_2$};
  \draw (14,4.05) node[right] {$x_3$};
\end{tikzpicture}
    \end{center}
    \caption{
      \label{fig:first_example}
      Illustration of Example~\ref{ex:first_example};
       $\pi_1=(3,2,1)$, $\pi_2=(1,2,3)$, $\pi_3=(3,1,2)$.
    }
  \end{figure}

  \vspace*{-0.2cm}
  \begin{definition}
    A \emph{basic interval} is an interval
    of the form $[i,i+1]$ (for $i=1,\ldots,k-1$) or $[i,i]$ (for $i=1,\ldots,k$).
    The former is called \emph{proper}, and the latter \emph{degenerate}.
    An \emph{interval system} is a sequence $\I=(I_1,\ldots,I_\ell)$
    of basic intervals $I_j$.
  \end{definition}

 For a basic interval $I_j$ we say that a value $x_j$ is consistent with $I_j$
 if $x_j\in \{s_{i,j},\ldots,s_{i+1,j}\}$ when $I_j=[i,i+1]$ is proper, and if $x_j = s_{i,j}$ when $I_j=[i,i]$
 is degenerate.
  A sequence $\x$ is called consistent with an interval system $\I=(I_j)_{j=1}^\ell$
  if for each $j$ the value $x_j$ is consistent with $I_j$.

  For an interval system $\I$ we define $\OPT(\A,\I)$
  as the minimum $\dist(\x,\A)$ among all integer sequences $\x$ consistent with $\I$.
  Due to the following trivial observation, for every $\A$ there exists an interval system $\I$
  such that $\OPT(\A) = \OPT(\A,\I)$.
     \begin{observation}\label{obs:bet}
    If $\x$ is a Manhattan consensus sequence then for each $j$,
    $s_{1,j} \le x_j \le s_{k,j}$.
  \end{observation}

  \paragraph{Transformation of the input to an ILP.}
  Note that for all sequences $\x$ consistent with a fixed $\I$, the Manhattan distances
  $\dist(\x, \a_i)$ can be expressed as $d_i+\sum_{j=1}^\ell e_{i,j}x_j$ with $e_{i,j}=\pm 1$.
  Thus, the problem of finding $\OPT(\A,\I)$ can be formulated as an ILP, which we denote $\ILP(\I)$.
  If $I_j$ is a proper interval $[i,i+1]$, we introduce a variable $x_j\in \{s_{i,j},\ldots,s_{i+1,j}\}$.
  Otherwise we do not need a variable $x_j$.
  The $i$-th constraint of $\ILP(\I)$ algebraically represents $\dist(\x, \a_i)$,
  see Example~\ref{ex:big_example}.
  \begin{observation}
    The optimal value of $\ILP(\I)$ is equal to $\OPT(\A, \I)$.
  \end{observation}
      \begin{figure}[bhtp]
    \begin{center}
      \tikzset{
    every node/.style={
        inner sep=2
    }
    }
    
\begin{tikzpicture}[scale=0.45,yscale=-.99,xscale=1.3]

\draw (0,11) node[circle,draw] (N11) {\scriptsize 2};
\draw (0,11) node[left=0.3cm] {11};
\draw (0,14) node[circle,draw] (N21) {\scriptsize 3};
\draw (0,14) node[left=0.3cm] {14};
\draw (0,16) node[circle,draw] (N31) {\scriptsize 5};
\draw (0,16) node[left=0.3cm] {16};
\draw (0,19) node[circle,draw] (N41) {\scriptsize 4};
\draw (0,19) node[left=0.3cm] {19};
\draw (0,20) node[circle,draw] (N51) {\scriptsize 1};
\draw (0,20) node[left=0.3cm] {20};
\draw (3,6) node[circle,draw] (N12) {\scriptsize 2};
\draw (3,6) node[left=0.3cm] {6};
\draw (3,8) node[circle,draw] (N22) {\scriptsize 4};
\draw (3,8) node[left=0.3cm] {8};
\draw (3,12) node[circle,draw] (N32) {\scriptsize 3};
\draw (3,12) node[left=0.3cm] {12};
\draw (3,15) node[circle,draw] (N42) {\scriptsize 5};
\draw (3,15) node[left=0.3cm] {15};
\draw (3,18) node[circle,draw] (N52) {\scriptsize 1};
\draw (3,18) node[left=0.3cm] {18};
\draw (6,7) node[circle,draw] (N13) {\scriptsize 2};
\draw (6,7) node[left=0.3cm] {7};
\draw (6,11) node[circle,draw] (N23) {\scriptsize 5};
\draw (6,11) node[left=0.3cm] {11};
\draw (6,16) node[circle,draw] (N33) {\scriptsize 4};
\draw (6,16) node[left=0.3cm] {16};
\draw (6,18) node[circle,draw] (N43) {\scriptsize 3};
\draw (6,18) node[left=0.3cm] {18};
\draw (6,20) node[circle,draw] (N53) {\scriptsize 1};
\draw (6,20) node[left=0.3cm] {20};
\draw (9,10) node[circle,draw] (N14) {\scriptsize 1};
\draw (9,10) node[left=0.3cm] {10};
\draw (9,13) node[circle,draw] (N24) {\scriptsize 3};
\draw (9,13) node[left=0.3cm] {13};
\draw (9,15) node[circle,draw] (N34) {\scriptsize 5};
\draw (9,15) node[left=0.3cm] {15};
\draw (9,17) node[circle,draw] (N44) {\scriptsize 2};
\draw (9,17) node[left=0.3cm] {17};
\draw (9,18) node[circle,draw] (N54) {\scriptsize 4};
\draw (9,18) node[left=0.3cm] {18};
\draw (12,6) node[circle,draw] (N15) {\scriptsize 5};
\draw (12,6) node[left=0.3cm] {6};
\draw (12,11) node[circle,draw] (N25) {\scriptsize 3};
\draw (12,11) node[left=0.3cm] {11};
\draw (12,12) node[circle,draw] (N35) {\scriptsize 4};
\draw (12,12) node[left=0.3cm] {12};
\draw (12,14) node[circle,draw] (N45) {\scriptsize 2};
\draw (12,14) node[left=0.3cm] {14};
\draw (12,16) node[circle,draw] (N55) {\scriptsize 1};
\draw (12,16) node[left=0.3cm] {16};
\draw (15,6) node[circle,draw] (N16) {\scriptsize 3};
\draw (15,6) node[left=0.3cm] {6};
\draw (15,8) node[circle,draw] (N26) {\scriptsize 1};
\draw (15,8) node[left=0.3cm] {8};
\draw (15,14) node[ultra thick,circle,draw] (N36) {\scriptsize 2};
\draw (15,14) node[left=0.3cm] {14};
\draw (15,17) node[circle,draw] (N46) {\scriptsize 5};
\draw (15,17) node[left=0.3cm] {17};
\draw (15,19) node[circle,draw] (N56) {\scriptsize 4};
\draw (15,19) node[left=0.3cm] {19};
\draw (18,10) node[circle,draw] (N17) {\scriptsize 1};
\draw (18,10) node[left=0.3cm] {10};
\draw (18,11) node[circle,draw] (N27) {\scriptsize 5};
\draw (18,11) node[left=0.3cm] {11};
\draw (18,12) node[circle,draw] (N37) {\scriptsize 3};
\draw (18,12) node[left=0.3cm] {12};
\draw (18,17) node[circle,draw] (N47) {\scriptsize 2};
\draw (18,17) node[left=0.3cm] {17};
\draw (18,19) node[circle,draw] (N57) {\scriptsize 4};
\draw (18,19) node[left=0.3cm] {19};
\draw (N11) -- (N21);
\draw[ultra thick] (N21) -- node[right] {$I_1$} (N31);
\draw (N31) -- (N41);
\draw (N41) -- (N51);
\draw (N12) -- (N22);
\draw[ultra thick] (N22) -- node[right] {$I_2$} (N32);
\draw (N32) -- (N42);
\draw (N42) -- (N52);
\draw (N13) -- (N23);
\draw[ultra thick] (N23) -- node[right] {$I_3$} (N33);
\draw (N33) -- (N43);
\draw (N43) -- (N53);
\draw (N14) -- (N24);
\draw (N24) -- (N34);
\draw[ultra thick] (N34) -- node[right] {$I_4$} (N44);
\draw (N44) -- (N54);
\draw (N15) -- (N25);
\draw (N25) -- (N35);
\draw[ultra thick] (N35) -- node[right] {$I_5$} (N45);
\draw (N45) -- (N55);
\draw (N16) -- (N26);
\draw (N26) -- (N36);
\draw (N36) -- (N46);
\draw (N46) -- (N56);
\draw (N17) -- (N27);
\draw (N27) -- (N37);
\draw[ultra thick] (N37) -- node[right] {$I_7$} (N47);
\draw (N47) -- (N57);
\draw (15,14) node[right=0.3cm] {$I_6$};
\end{tikzpicture}
    \end{center}
    \caption{
      \label{fig:big_example}
      Illustration of Example~\ref{ex:big_example}: 5 sequences of length 7
      together with an interval system.
      Notice that $I_6$ is a degenerate interval.
    }
  \end{figure}
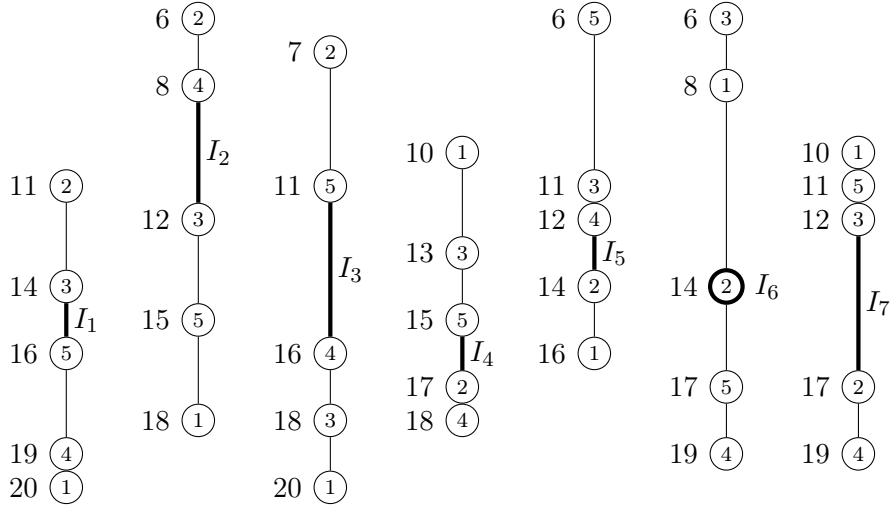

  \vspace*{-0.2cm}
  \begin{example}\label{ex:big_example}
    Consider the following 5 sequences of length 7:
    $$[a_{i,j}] =
      \begin{bmatrix}
        20 & 18 & 20 & 10 & 16 &  8 & 10\\
        11 &  6 &  7 & 17 & 14 & 14 & 17\\
        14 & 12 & 18 & 13 & 11 &  6 & 12\\
        19 &  8 & 16 & 18 & 12 & 19 & 19\\
        16 & 15 & 11 & 15 &  6 & 17 & 11
      \end{bmatrix}
    $$
    and an interval system $\I = ([2,3],\;[2,3],\;[2,3],\;[3,4],\;[3,4],\;[3,3],\;[3,4])$.
    An illustration of both can be found in Fig.~\ref{fig:big_example}.

    We obtain the following $\ILP(\I)$, where
    $x_1 \in [14,16]$,
    $x_2 \in [8,12]$,
    $x_3 \in [11,16]$,
    $x_4 \in [15,17]$,
    $x_5 \in [12,14]$,
    $x_7 \in [12,17]$
    and the sequence $\x$ can be retrieved as $\x=(x_1,x_2,x_3,x_4,x_5,14,x_7)$:
  
  {\small
    $$
    \begin{matrix}
                 &       &        &       &        &       &        &       &\min z  &       &   &       &        & \\
          20-x_1 & \ \ + & 18-x_2 & \ \ + & 20-x_3 & \ \ + & x_4-10 & \ \ + & 16-x_5 & \ \ + & 6 & \ \ + & x_7-10 & \;\;\le \;z\\
          x_1-11 & \ \ + & x_2-6 & \ \ + & x_3-7 & \ \ + & 17-x_4 & \ \ + & 14-x_5 & \ \ + & 0 & \ \ + & 17-x_7 & \;\;\le \;z\\
          x_1-14 & \ \ + & 12-x_2 & \ \ + & 18-x_3 & \ \ + & x_4-13 & \ \ + & x_5-11 & \ \ + & 8 & \ \ + & x_7-12 & \;\;\le \;z\\
          19-x_1 & \ \ + & x_2-8 & \ \ + & 16-x_3 & \ \ + & 18-x_4 & \ \ + & x_5-12 & \ \ + & 5 & \ \ + & 19-x_7 & \;\;\le \;z\\
          16-x_1 & \ \ + & 15-x_2 & \ \ + & x_3-11 & \ \ + & x_4-15 & \ \ + & x_5-6 & \ \ + & 3 & \ \ + & x_7-11 & \;\;\le \;z
    \end{matrix}
    $$
  }

  \end{example}
     Note that $P=\ILP(\I)$ has the following special form, which we call \pmilp:
  \begin{align*}
  \min z\\
  d_i + \sum_{j} x_j e_{i,j} \le z\\
  x_j \in R_P(x_j)
  \end{align*}
  where $e_{i,j}=\pm 1$ and $R_P(x_j)=\{\ell_j,\ldots,r_j\}$ for integers $\ell_j\le r_j$.
  Whenever we refer to variables, it does not apply to $z$, which is of auxiliary character.
  Also, ``$x_{j}\in R_P(x_j)$'' are called variable ranges rather than constraints. 
  We say that $(e_{1,j},\ldots,e_{k,j})$  is a \emph{coefficient vector} of $x_j$ and denote it as $\E_P(x_j)$.
  If the program $P$ is apparent from the context, we omit the subscript.

      \paragraph{Simplification of ILP.}
  The following two facts are used to reduce the number of variables of a \pmilp.
  For $A,B\sub \mathbb{Z}$ we define $-A=\{-a : a\in A\}$ and  $A+B = \{a+b : a\in A, b\in B\}$.

    \begin{fact}\label{fct:negate}
    Let $P$ be a \pmilp. Let $P'$ be a program obtained from $P$ by replacing
    a variable $x_j$ with $-x_j$, i.e. setting $\E_{P'}(x_j)=-\E_P(x_j)$ and $R_{P'}(x_j) = -R_P(x_j)$.
    Then $\OPT(P)=\OPT(P')$.
  \end{fact}

  \begin{fact}\label{fct:join}
  Let $P$ be a \pmilp. Assume $E_P(x_j)=E_P(x_{j'})$ for $j\ne j'$. Let $P'$ be a program obtained from $P$ by removing the variable
   $x_{j'}$ and replacing $x_j$ with $x_{j}+x_{j'}$, i.e. setting $R_{P'}(x_j) = R_P(x_j)+R_{P}(x_{j'})$.
    Then $\OPT(P)=\OPT(P')$.
  \end{fact}
  \begin{proof}
  	Let $(z,x_1,\ldots,x_n)$ be a feasible solution of $P$. Then setting $x_{j}:=x_{j}+x_{j'}$
  	and removing the variable $x_{j'}$ we obtain a feasible solution of $P'$. Therefore $\OPT(P')\le \OPT(P)$.
  	For the proof of the other inequality, take a feasible solution $(z,x_1,\ldots,x_n)$ (with $x_{j'}$ missing)
  	of $P'$. Note that $x_{j}\in R_{P'}(x_j)=R_P(x_j)+R_P(x_{j'})$. Therefore one can split
  	$x_j$ into $x_j + x_{j'}$ so that $x_j\in R_P(x_j)$ and $x_{j'}\in R_P(x_{j'})$. This way we obtain a feasible solution of $P$
  	and thus prove that $\OPT(P)\le \OPT(P')$.
  \end{proof}

  \begin{corollary}\label{cor:2_to_k}
    For a \pmilp with $k$ constraints one can compute in linear time an equivalent \pmilp
    with $k$ constraints and up to $2^{k-1}$ variables.
  \end{corollary}
  \begin{proof}
  We apply Fact~\ref{fct:negate} to obtain $e_{1,1}=e_{1,2}=\ldots=e_{1,\ell}$,
  this leaves at most $2^{k-1}$ different coefficient vectors.
  Afterwards we apply Fact~\ref{fct:join} as many times as possible so that there is
  exactly one variable with each coefficient vector.
\end{proof}

  \begin{example}
    Consider the \pmilp $P$ from Example~\ref{ex:big_example}.
    Observe that $\E_P(x_4) = \E_P(x_7) = -\E_P(x_2)$ and thus Facts~\ref{fct:negate} and~\ref{fct:join}
    let us merge $x_2$ and $x_7$ into $x_4$ with
    $$R_{P'}(x_4) = R_P(x_4) + R_P(x_7) - R_P(x_2) = [15,17] + [12,17] + [-12,-8] = [15,26].$$
    Simplifying the constant terms we obtain the following \pmilp $P'$:

    {\small
      $$
    \begin{matrix}
      &&&&&&&&&\min z\\
-x_1  & - & x_3 & + & x_4 & - & x_5 & + & 60 & \le z\\
+x_1  & + & x_3 & - & x_4 & - & x_5 & + & 24 & \le z\\
+x_1  & - & x_3 & + & x_4 & + & x_5 & - & 12 & \le z\\
-x_1 & - & x_3 & - & x_4 & + & x_5 & + & 57 & \le z\\
-x_1 & + & x_3 & + & x_4 & + & x_5 & - & 9  & \le z
    \end{matrix}
    $$
  }

  \end{example}

\section{Kernel of MSC for Arbitrary $k$}\label{sec:ker}
In this section we give a kernel for the MSC problem parameterized with $k$,
which we then apply to develop a linear-time FPT algorithm.
To obtain the kernel we need a combinatorial observation that if $\pi_{j} = \pi_{j'}$
then the $j$-th and the $j'$-th column in $\A$ can be merged.
This is stated formally in the following lemma.

\begin{lemma}\label{lem:kernel_join}
  Let $\A = (\a_1,\ldots,\a_k)$ be a collection of sequences of length $\ell$
  and assume that $\pi_{j} = \pi_{j'}$ for some $1 \le j < j' \le \ell$.
  Let $\A' = (\a'_1,\ldots,\a'_k)$ be a collection of sequences of length $\ell-1$
  obtained from $\A$ by removing the $j'$-th column and setting
  $a'_{i,j} = a_{i,j}+a_{i,j'}$.
  Then $\OPT(\A) = \OPT(\A')$.
\end{lemma}
\begin{proof}
	First, let us show that $\OPT(\A')\le \OPT(\A)$.
	Let $\x$ be a Manhattan consensus sequence for $\A$
  and let $\x'$ be obtained from $\x$ by removing the $j'$-th entry and setting $x'_j=x_{j}+x_{j'}$.
	We claim that $\dist(\x', \A')\le \dist(\x, \A)$.
	Note that it suffices to show that $|x'_j-a'_{i,j}|\le |x_j-a_{i,j}|+|x_{j'}-a_{i,j'}|$ for all $i$.
	However, with $x'_j = x_{j}+x_{j'}$ and $a'_{i,j}=a_{i,j}+a_{i,j'}$, this is a direct consequence
	of the triangle inequality.

	It remains to prove that $\OPT(\A)\le \OPT(\A')$.
	Let $\x'$ be a Manhattan consensus sequence for $\A'$.
	By Observation~\ref{obs:bet}, $x'_{j}$ is consistent with some proper basic interval $[i,i+1]$.
	Let $d'_{i,j}= x'_j - s'_{i,j}$ and $D'_{i,j}=s'_{i+1,j}-s'_{i,j}$.
	Also, let $D_{i,j}=s_{i+1,j}-s_{i,j}$ and $D_{i,j'}=s_{i+1,j'}-s_{i,j'}$.
  Note that, since $\pi_j = \pi_{j'}$, $D'_{i,j}= D_{i,j} + D_{i,j'}$.
  Thus, one can partition $d'_{i,j}=d_{i,j}+d_{i,j'}$
	so that both $d_{i,j}$ and $d_{i,j'}$ are non-negative integers not exceeding $D_{i,j}$ and $D_{i,j'}$ respectively.
	We set $x_j = s_{i,j}+d_{i,j}$ and $x_{j'}=s_{i,j'}+d_{i,j'}$, and the remaining components of $\x$
  correspond to components of $\x'$.
	Note that $x'_j = x_j+x_{j'}$ and that both $x_j$ and $x_{j'}$ are consistent with $[i,i+1]$.
	Consequently for any sequence $\a_m$ it holds that $\dist(\x,\a_m)=\dist(\x',\a_m)$ and therefore $\dist(\x,\A)=\dist(\x',\A')$,
	which concludes the proof.

	Finally, note that the procedures given above can be used to efficiently convert between the optimum solutions $\x$ and $\x'$.
\end{proof}

\noindent
By Lemma~\ref{lem:kernel_join}, to obtain the desired kernel we need to sort the elements in columns of $\A$
and afterwards sort the resulting permutations $\pi_j$:

\begin{theorem}\label{thm:kernel}
  In $\Oh(\ell k\log k)$ time one can reduce any instance of MSC to
  an instance with $k$ sequences of length $\ell'$, with $\ell' \le k!$.
\end{theorem}
\begin{remark}
For binary instances, if permutations $\pi_j$ are chosen appropriately,
 it holds that $\ell'\le 2^k$.
\end{remark}

\begin{theorem}\label{thm:impractical}
 For any integer $k$, the \textsc{Manhattan Sequence Consensus} problem can be solved in $\Oh(\ell k\log k+ 2^{k!\log k+ O(k2^k)}\log M)$ time.
\end{theorem}
\begin{proof}
  We solve the kernel from Theorem~\ref{thm:kernel} by considering all possible interval systems $\I$
  composed of proper intervals.
  The sequences in the kernel have length at most $k!$, which gives $(k-1)^{k!}$
  {\pmilp}s of the form $\ILP(\I)$ to solve.

  Each of the {\pmilp}s initially has $k$ constraints on $k!$ variables but, due to
  Corollary~\ref{cor:2_to_k}, the number of variables can be reduced to $2^{k-1}$.
  Lenstra's algorithm with further improvements \cite{kannan,franktardos,lokshtanov}
  solves ILP with $p$ variables in $\Oh(p^{2.5p+o(p)}\log L)$ time, where $L$ is the bound on
  the scope of variables.
  In our case $L=\Oh(\ell M)$, which  gives the time complexity of:
  $$\Oh\left((k-1)^{k!} \cdot 2^{(k-1)(2.5\cdot2^{k-1}+o(2^{k-1}))}\log M\right) = \Oh\left(2^{k!\log k + O(k2^{k})}\log M\right).$$
  This concludes the proof of the theorem.
\end{proof}

\section{Combinatorial Characterization of Solutions for $k = 5$}\label{sec:comb5}
  In this section we characterize
  those Manhattan consensus sequences $\x$ which additionally
  minimize $\sum_{i} \dist(\x,\a_i)$ among all Manhattan consensus sequences.
  Such sequences are called here \emph{sum-MSC sequences}.
  We show that one can determine a collection of 20 interval systems,
  so that any sum-MSC sequence is guaranteed to be consistent with one of them.
  We also prove some structural properties of these systems, which are then
  useful to efficiently solve the corresponding {\pmilp}s.

	We say that $x_j$ is in the \emph{center} if $x_j=s_{3,j}$, i.e. $x_j$ is equal to the column median.
	Note that if $x_j \ne s_{3,j}$, then moving $x_j$ by one towards the center
	decreases by one $\dist(\x,\a_i)$ for at least three sequences $\a_i$.

	\begin{definition}
	We say that $\a_i$ \emph{governs} $x_j$ if $x_j$ is in the center or moving $x_j$ towards the center
	increases $\dist(\x,\a_i)$. The set of indices $i$ such that $\a_i$ governs $x_j$ is denoted
	as $G_j(\x)$.
	\end{definition}
	Observe that if $x_j$ is in the center, then $|G_j(\x)|=5$, and otherwise
	$|G_j(\x)|\le 2$; see Fig.~\ref{fig:G}.
  For $k=5$ we have 4 proper basic intervals: $[1,2]$, $[2,3]$, $[3,4]$ and $[4,5]$.
  We call $[1,2]$ and $[4,5]$ \emph{border} intervals, and the other two \emph{middle intervals}.
  We define $G_j([1,2])=\{\pi_j(1)\}$, $G_j([2,3])=\{\pi_j(1),\pi_j(2)\}$, $G_j([3,4])=\{\pi_j(4),\pi_j(5)\}$
  and $G_j([4,5])=\{\pi_j(5)\}$.
  Note that if we know $G_j(\x)$ and $|G_j(\x)|\le 2$, then we are guaranteed that $x_j$ is consistent
  with the basic interval $I_j$ for which $G_j(I_j)=G_j(\x)$.

  Observe that if $\x$ is a Manhattan consensus sequence,
  then $G_j(\x)\ne \emptyset$ for any $j$.
  If we additionally assume that $\x$ is a sum-MSC sequence, we obtain a stronger property.

\begin{figure}[htpb]
  \begin{center}
    \begin{tikzpicture}[xscale=1.5]
      \draw (0,0) -- (6,0);
      \foreach \x in {1,...,3} \draw[xshift=-0.03cm] (\x,0) node {)};
      \foreach \x in {1,...,2} \draw[xshift=0.03cm] (\x,0) node {[};
      \foreach \x in {4,...,5} \draw[xshift=-0.03cm] (\x,0) node {]};
      \foreach \x in {3,...,5} \draw[xshift=0.03cm] (\x,0) node {(};
      \foreach \x in {1,...,5} \draw[above=0.2cm] (\x,0) node {\x};
      \draw (-0.2,0) node[left] {$x_j$:};

      \draw (0.5,-0.7) node {$\emptyset$};
      \draw (1.5,-0.7) node {$\{3\}$};
      \draw (2.5,-0.7) node {$\{1,3\}$};
      \draw (3.5,-0.7) node {$\{2,5\}$};
      \draw (4.5,-0.7) node {$\{5\}$};
      \draw (5.5,-0.7) node {$\emptyset$};
      \draw (3,-1.3) node (A) {$\{1,\ldots,5\}$};
      \draw[-latex] (A) -- (3,-0.2);
      \draw (-0.2,-0.7) node[left] {$G_j(\x)$:};
    \end{tikzpicture}
  \end{center}
  \caption{\label{fig:G}
    Assume $a_{1,j}=2, a_{2,j}=4, a_{3,j}=1, a_{4,j}=3,a_{5,j}=5$.
    Then $G_j(x)$ depends on the interval of $x_j$ as shown in the figure.
    E.g., if $1 \le x_j < 2$ then $G_j(\x)=\{3\}$.
  }
\end{figure}
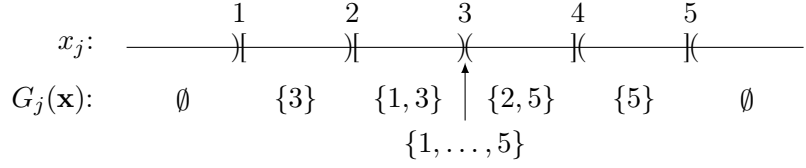


 \begin{lemma}\label{lem:com}
  Let $\x$ be a sum-MSC sequence.
 	Then $G_j(\x)\cap G_{j'}(\x)\ne \emptyset$ for any $j,j'$.
  \end{lemma}
  \begin{proof}
  For a proof by contradiction assume $G_j(\x)$ and $G_{j'}(\x)$ are disjoint.
  This implies that neither $x_j$ nor $x_{j'}$ is in the center and thus $|G_j(\x)|,|G_{j'}(\x)|\le 2$.
  Let us move both $x_j$ and $x_{j'}$ by one towards the center.
  Then $\dist(\x,\a_i)$ remains unchanged for $i\in G_j(\x)\cup G_{j'}(\x)$ (by disjointness),
  and decreases by two for the remaining sequences $\a_i$. There must be at least one
  such remaining sequence, which contradicts our choice of $\x$.
  \end{proof}
  Additionally, if a sum-MSC sequence $\x$ has a position $j$ with $|G_j(\x)|=1$,
  the structure of $\x$ needs to be even more regular.
  \begin{definition}
A sequence $\x$ is called an $i$-\emph{border sequence} if for each $j$ it holds that $x_j = a_{i,j}$ or $G_j(\x)=\{i\}$.
\end{definition}
  \begin{lemma}\label{lem:one}
  Let $\x$ be a sum-MSC sequence. If $G_j(\x)=\{i\}$ for some $j$, then $\x$ is an $i$-border sequence.
  \end{lemma}
	\begin{proof}
	  For a proof by contradiction assume $G_{j'}(\x)\ne \{i\}$ and $x_{j'}\ne a_{i,j'}$ for some $j'$.
	  Let us move $x_j$ towards the center and $x_{j'}$ towards $a_{i,j'}$ both by one.
	  Then for any $i'$ it holds that $\dist(\x,\a_{i'})$ does not increase.
	  By Lemma~\ref{lem:com} $i\in G_{j'}(\x)$, so $x_{j'}$ is moved away from the center.
	  Moreover, $x_j$ is moved towards some $a_{i',j}$ with $i'\ne i$, since $G_{j'}(\x)\ne \{i\}$.
	  Consequently, $\dist(\x,\a_{i'})$ decreases by two, which
	    contradicts our choice of $\x$. 
	\end{proof}

	\begin{definition}
	A sequence $\x$ is called an $i$-\emph{middle sequence} if for each $j$ it holds that $i\in G_j(\x)$ and $|G_j(\x)|\ge 2$.
	\end{definition}
	\begin{definition}
	For a 3-element set $\Delta\subseteq \{1,\ldots,5\}$ a sequence $\x$ is called a $\Delta$-\emph{triangle sequence}
	if for each $j$ it holds that $|\Delta\cap G_j(\x)|\ge 2$.
	\end{definition}

	\begin{lemma}\label{lem:char}
	Let $\x$ be a sum-MSC sequence.
	Then $\x$ is a border sequence, a middle sequence or a triangle sequence.
	\end{lemma}
	\begin{proof}
  Recall that $G_j(\x)\ne \emptyset$ for each $j$.
	By Lemma~\ref{lem:one}, if $G_j(\x)=\{i\}$ for some $j$, then $\x$ is an $i$-border sequence.
	This lets us assume $|G_j(x)|\ge 2$, i.e. $|G_j(\x)|\in\{2,5\}$, for each $j$.
	Let $\F$ be the family of 2-element sets among $G_j(\x)$. By Lemma~\ref{lem:com} every two of them intersect,
	so we can apply the following easy set-theoretical claim.

	\begin{claim}
	Let $\G$ be a family of 2-element sets such that every two sets in $\G$ intersect.
	 Then sets in $\G$ share a common element
	or $\G$ contains exactly three sets with three elements in total.
	\end{claim}

	If all sets in $\F$ share an element $i$, then $\x$ is clearly an $i$-middle sequence.
	Otherwise $\x$ is a $\Delta$-triangle sequence for $\Delta=\bigcup \F$.
	\end{proof}

	\begin{fact}\label{fct:is}
	There exist 20 interval systems $\B_i,\M_i$ (for $i\in\{1,\ldots,5\}$)
	and $\T_{\Delta}$ (for 3-element sets $\Delta \sub \{1,\ldots,5\}$)
	such that:\vspace{-.2cm}
	\begin{enumerate}[(a)]
	  \item\label{it:border} $\B_i$ is consistent with all $i$-border sequences, $G_j(\B_{i,j})=\{i\}$ for proper $\B_{i,j}$;
	\item\label{it:middle} $\M_i$ is consistent with all $i$-middle sequences and if $\M_{i,j}$ is proper then $|G_j(\M_{i,j})|=2$ and $i\in G_j(\M_{i,j})$;
	\item\label{it:triangle} $\T_{\Delta}$ is consistent with all $\Delta$-triangle sequences and if  $\T_{\Delta,j}$ is proper then $|G_j(\T_{\Delta,_j})|=2$ and $G_j(\T_{\Delta,j})\sub\Delta$.
	\end{enumerate}
	\end{fact}
	\begin{proof}
	(\ref{it:border}) Let us fix a position $j$. For any $i$-border sequence $\x$, we know that $x_j = a_{i,j}$ or $G_{j}(\x)=\{i\}$.
	If either of the border intervals $I$ satisfies $G_j(I)=\{i\}$, we set $\B_{i,j}:=I$ (observe
	that $a_{i,j}$ is then consistent with $I$). Otherwise we choose $\B_{i,j}$
	so that it is degenerate and corresponds to $x_{j}=a_{i,j}$.

	\noindent
	(\ref{it:middle}) Again fix $j$. For any $i$-middle sequence $\x$, we know that $x_j$ is consistent with at least one of the two middle intervals
	(both if $x_j$ is in the center). If either of the middle intervals $I$ satisfies $i\in G_j(I)$, we choose $\M_{i,j}:=I$. (Note that this condition
	cannot hold for both middle intervals). Otherwise we know that $x_{j}$ is in the center and set $\M_{i,j}$
	so that it is degenerate and corresponds to $x_{j}$ in the center, i.e. $\M_{i,j}:=[3,3]$.

	\noindent
	(\ref{it:triangle}) We act as in (\ref{it:middle}), i.e. if either of the middle intervals $I$ satisfies $|G_j(I)\cap \Delta|=2$, we choose $\T_{\Delta,j}:=I$
	(because sets $G_{j}(I)$ are disjoint for both middle intervals, this condition cannot hold for both of them).
	Otherwise, we set $\T_{\Delta,j}:=[3,3]$, since $x_{j}$ is guaranteed to be in the center for any $\Delta$-triangle sequence $\x$.
	\end{proof}

  \section{Practical Algorithm for $k \le 5$}\label{sec:algo5}
  It suffices to consider $k=5$.
  Using Fact~\ref{fct:is} we reduce the number of interval systems from $(k-1)^{k!}=4^{5!}>10^{72}$ to $20$
  compared to the algorithm of Section~\ref{sec:ker}.
  Moreover, for each of them $\ILP(\I)$ admits structural properties,
  which lets us compute $\OPT(\A,\I)$ much more efficiently than using a general ILP solver.

 \begin{definition}
 A \pmilp is called \emph{easy} if for each constraint the number of $+1$ coefficients is $0$, $1$ or $n$,
 where $n$ is the number of variables.
 \end{definition}

 \begin{lemma}\label{lem:reduce}
For each $\I$  being one of the 20 interval systems $\B_i,\M_i$ and $\T_{\Delta}$,
$\ILP(\I)$ can be reduced to an equivalent easy \pmilp with up to 4 variables.
\end{lemma}
\begin{proof}
Recall that for degenerate intervals $I_j$, we do not introduce variables.
On the other hand, if $I_j$ is proper, possibly negating the variable $x_j$ (Fact~\ref{fct:negate}),
we can make sure that the coefficient vector $\E(x_j)$ has $+1$ entries corresponding to $i\in G_j(I_j)$ and $-1$ entries for
the remaining $i$. Moreover, merging the variables (Fact~\ref{fct:join}), we end up with
a single variable per possible value $G_j(I_j)$. Now we use structural properties stated in Fact~\ref{fct:is}
to claim that the \pmilp we obtain this way, possibly after further variable negations, becomes easy.

\textbf{Border sequences.}
By Fact~\ref{fct:is}(\ref{it:border}), if $\B_{i,j}$ is proper, then $G_j(\B_{i,j})=\{i\}$
and thus the \pmilp has at most 1 variable and consequently is easy.

\textbf{Middle sequences.}
By Fact~\ref{fct:is}(\ref{it:middle}), if $\M_{i,j}$ is proper, then $G_j(\M_{i,j})=\{i,i'\}$
for some $i'\ne i$. Thus there are up to 4 variables, the constraint corresponding to $i$ has only $+1$ coefficients,
and the remaining constraints have at most one $+1$.

\textbf{Triangle sequences.}
By Fact~\ref{fct:is}(\ref{it:triangle}), if $\T_{\Delta,j}$ is proper, then $G_j(\T_{\Delta,j})$ is a
2-element subset of $\Delta$, and thus there are up to three variables.
Any \pmilp with up to two variables is easy,
and if we obtain three variables, then the constraints corresponding to $i\in \Delta$
have exactly two $+1$ coefficients, while the constraints corresponding to $i\notin \Delta$ have just $-1$ coefficients.
Now, negating each variable (Fact~\ref{fct:negate}), we get one $+1$ coefficient in constraints corresponding to $i\in \Delta$
and all $+1$ coefficients for $i\notin \Delta$.
\end{proof}

\noindent
The algorithm of Lenstra~\cite{Lenstra} with further improvements
\cite{kannan,franktardos,lokshtanov}, which runs in roughly $n^{2.5n + o(n)}$ time,
could perform reasonably well for $n=4$.
However, there is a simple $\Oh(n^2)$-time algorithm designed for easy \pmilp.
It can be found in Section~\ref{app:ilp}.

In conclusion, the algorithm for MSC problem first proceeds as described in Fact~\ref{fct:is} to obtain the interval systems $\B_i,\M_i$ and $\T_{\Delta}$.
For each of them it computes $\ILP(\I)$, as described in Section~\ref{sec:prelim},
and converts it to an equivalent easy \pmilp following Lemma~\ref{lem:reduce}.
Finally, it uses the efficient algorithm to solve each of these 20 \pmilp{s}.
The final result is the minimum of the optima obtained.

\section{Conclusions}
We have presented an $\Oh(\ell k \log k)$-time kernelization algorithm,
which for any instance of the MSC problem computes an equivalent
instance with $\ell' \le k!$.
Although for $k \le 5$ this gives an instance with $\ell'\le 120$, 
i.e.\ the kernel size is constant, 
solving it in a practically feasible time remains challenging.
Therefore for $k \le 5$ we have designed an efficient linear-time algorithm.

We have implemented the algorithm,
including retrieving the optimum consensus sequence (omitted in the description above) \footnote{
It is available at \url{http://www.mimuw.edu.pl/~kociumaka/files/msc.cpp}.
}.
For random input data with $\ell=10^6$ and $k=5$, the algorithm without kernelization achieved
the running time of 1.48015s, which is roughly twice the time required to read the data (0.73443s, not included in the former).
The algorithm pipelined with the kernelization achieved 0.33415s.
The experiments were conducted on a MacBook Pro notebook (2.3~Ghz Intel Core i7, 8 GB RAM).


\nopagebreak

  \bibliographystyle{abbrv}
  \bibliography{cs.bib}

\newpage

\appendix
\section{Solving easy \pmilp}\label{app:ilp}
  Recall the definition of an easy \pmilp:
  \begin{align*}
  \min z\\
  d_i + \sum_{j} x_j e_{i,j} \le z\\
  x_j \in R(x_j)
  \end{align*}
  where $e_{i,j}=\pm 1$,
  and for each constraint the number of $+1$ coefficients is $0$, $1$ or $n$,
  where $n$ is the number of variables.
  Moreover $R(x_j)=\{\ell_j,\ldots,r_j\}$ for integers $\ell_j\le r_j$.

  We define a \emph{canonical} form of easy \pmilp as follows:
    \begin{align*}
     \min\quad  &  z\\
      y_1+y_2+\ldots+ y_n +K' & \le z\\
      y_1-y_2- \ldots -y_n + K_1 & \le z\\
      -y_1+y_2- \ldots -y_n + K_2 & \le z\\
      \vdots \\
      -y_1-y_2 -\ldots +y_n + K_n & \le z \\
      y_i &\in \{0,\ldots,D_i\}
    \end{align*}
    where $D_i\in \pint$ and $K',K_1,\ldots,K_n$ are of the same parity.

    \begin{lemma}\label{lem:red}
    Any easy \pmilp with $n$ variables and at least one constraint
    can be reduced to two easy \pmilp{s} in the canonical form,
    with $n+1$ variables, so that the optimum value of the input \pmilp
    is equal to the smaller optimum value of the two output \pmilp{s}.
    \end{lemma}
    \begin{proof}
    First, observe that the substitution $x_j = y_j+\ell_j$
    affects only the constant terms and the ranges of variables, and
    gives a \pmilp where the range of $y_j$ is $\{0,\ldots,D_j\}$ for $D_j = r_j - \ell_j$.

    Recall that the number of $+1$ coefficients in a constraint of an easy \pmilp
    is 0, 1 or $n$.
    Note that if two constraints have equal coefficients for each variable,
    then the constraint with smaller constant term is weaker than the other, and thus it can be removed.
    Consequently, there are at most $n+2$ constraints: one (denoted $C_+$) only with coefficients $+1$,
    one (denoted $C_-$) only with coefficients $-1$, and for each variable $y_j$ one constraint (denoted $C_j$) with coefficient $+1$ for $y_j$ and 	$-1$ for the remaining variables.
    Also, unless there are no constraints at all, if some of these $n+2$ constraints
    are missing, one may add it, setting the constant term to a sufficiently small number,
    e.g.\ $K-2\sum_{i}D_i$, where $K$ is the constant term of any existing constraint $C$.
    Then, with such a constant term, the introduced constraint follows from the existing one.

    Finally, we introduce a dummy variable $y_{n+1}$ with $D_{n+1}=0$ and coefficient $+1$ in $C_-$ and $C_+$, and
    coefficient $-1$
    in the constraints $C_j$.
    This way $C_-$ can be interpreted as $C_{n+1}$ and we obtain an \pmilp in the canonical form, which
    does not satisfy the parity condition yet.

   Now, we replace the \pmilp with two copies, in one of them we increment (by one)
   each odd constant term, and in the other we increment (by one) each even constant term.
   Clearly, both these \pmilp{s} are stronger than the original \pmilp. Nevertheless
   we claim that the originally feasible solution remains feasible for one of them.

   Consider a feasible solution $(z,y_1,\ldots,y_{n+1})$ and the parity of $P = z+y_1+\ldots+y_{n+1}$.
    Note that changing a $+1$ coefficient to a $-1$ coefficient does not change the parity,
    so only constraints with constant terms of parity $P\bmod 2$ may be tight.
    For the remaining ones, we can increment the constant term without violating the feasibility.
    This is exactly how one of the output \pmilp{s} is produced.
    This completes the proof of the lemma.
    \end{proof}

    \SetKwFunction{FSolve}{Solve}

    In the following we concentrate on solving easy {\pmilp}s in the canonical form.
    Moreover, we assume that $K_1\le \ldots \le K_n$ (constraints and variables can be renumbered
    accordingly).
    We can solve such ILP using function \FSolve.
    This function is quite inefficient, we improve it later on.
    An explanation of the pseudocode is contained in the proof of the following lemma.

  \begin{function}
 \caption{Solve($K',K_1,\ldots,K_n, D_1,\ldots,D_n$)}
 \lIf{$n=1$}{\Return $\max(K',K_1)$}
 \While{$K'<K_n$}{
 $\ell := \max\{i : K_i = K_1\}$\;
 \For{$i=1$ \KwSty{to} $\ell$}{
 	\If{$D_i =0$}{\Return{$\FSolve(K',K_1,\ldots,K_{i-1},K_{i+1},\ldots,K_n, D_1,\ldots,D_{i-1},D_{i+1},\ldots,D_n)$}\;}
 }
	$K' := K'+1$\;
	\lFor{$i=1$ \KwSty{to} $\ell-1$}{$K_i := K_i-1$}
	$K_{\ell}:= K_\ell+1$\;
	$D_{\ell}:= D_{\ell}-1$\;
	\lFor{$i=\ell+1$ \KwSty{to} $n$}{$K_i := K_i-1$}
}
\Return{$K'$}\;
  \end{function}

    \begin{lemma}\label{lem:solve}
    \FSolve correctly determines the optimum value of the easy \pmilp in the canonical form,
    provided that $K_1\le \ldots \le K_n$.
    \end{lemma}
    \begin{proof}
    If $n=1$ then the constraints are $y_1+K'\le z$ and $y_1+K_1\le z$, therefore the optimum is clearly $y_1=0$
    and $z=\min(K',K_1)$. Consequently line 1 is correct and in the following we may assume that $n\ge 2$.

    Due to the constraint $\sum_{i} y_i + K'\le z$ and the fact that $y_i$ are non-negative,
    the value $K'$ is clearly a lower bound on $z$. Moreover, if $K_n \le K'$ then taking $y_1=\ldots=y_n=0$
    and $z=K'$ is feasible, since in that case all the constraints reduce to $K_i \le K'$
    which is true due to monotonicity of $K_i$.
    Thus, we may exit the while-loop and return $K'$ in line 12 when $K_n \le K'$.

    Let us define $\ell$ as in line 3.
    If $D_i=0$ for some $1 \le i \le \ell$, then $y_i$ must be equal to $0$, and thus this variable can be removed. Then the constraint $C_i$ becomes of the form
    $-\sum_{j\ne i}y_j + K_i \le z$.
    However, since $n\ge 2$, there is another constraint $C_{i'}$
    of the form
    $y_{i'}-\sum_{j \ne i'}y_j+K_{i'} \le z$, which becomes
    $2y_{i'} -\sum_{j\ne i}y_j + K_{i'}\le z$ after $y_i$ is removed.
    Because of the inequality $K_{i'}\ge K_i$ and due to non-negativity
    of $y_{i'}$, $C_{i'}$ gives a stronger lower bound on $z$ than $C_i$, i.e.\ the constraint $C_i$ can be removed.
    Thus lines 3-6 are correct. Also, note that in the recursive call the monotonicity assumption still holds,
    also constant terms are still of the same parity.

    Consequently, in lines 7-11 we may assume that $D_i>0$ for $i\le \ell$.
    The arithmetic operations in these lines correspond to replacing
    $y_{\ell}$ with $y_{\ell}+1\in \{1,\ldots,D_{\ell}\}$. This operation is valid
    provided that the following claim holds.

    \begin{claim}
    There exists an optimal solution $(z,y_1,\ldots,y_n)$ with $y_{\ell}\ge 1$.
    \end{claim}
    \begin{proof}
    Let us take an optimal solution $(z,y_1,\ldots,y_n)$. If $y_\ell\ge 1$ we are done, so assume $y_{\ell}=0$.
    We consider several cases.
    First, assume that there exists some $i$ with $y_i\ge 1$.
    Replace $y_{i}$ with $y_{i}-1$ and set $y_{\ell}$ to $1$. By $D_{\ell}\ge 1$ this is feasible
    with respect to variable ranges. The only constraints that change are $C_\ell$ and $C_i$.
    Set $S=\sum_j y_j$, which also does not change. Note that $C_{\ell}$ changes from $z\ge K_{\ell}-S$ to $z\ge K_{\ell}-S+2$
    and $C_{i}$ changes from $z\ge K_{i}-S+2y_{i}$ to $z\ge K_{i}-S+2y_{i}-2$. Thus
    $C_i$ only weakens and the new $C_{\ell}$ is weaker than the original $C_i$, since $y_{i}\ge 1$ and $K_i\ge K_{\ell}$.
    Consequently, unless $y_1=\ldots=y_n=0$ is the only optimum, there exists an optimum with $y_{\ell}=1$.

    Thus, assume $y_1=\ldots = y_n = 0$ is an optimum. Since $K'<K_n$, we have $z=K_n$ for the optimum.
    Now, we consider two cases. First, assume $\ell<n$.
    We claim that after setting $y_\ell:=1$, the solution is still feasible.
    The only constraints that become stronger are $C_+$ and $C_{\ell}$. For the former we need to show
    that $K'+1\le K_n$, which is satisfied since $K'<K_n$. For the latter we need $K_\ell+1\le K_n$,
    which is also true since, by $\ell<n$, we have $K_{\ell}<K_n$.

    Now, assume $\ell=n$. By $n\ge 2$ this means that $\ell>1$.
    Therefore we can increment both $y_1$ and $y_{\ell}$ to 1 and observe that only $C_+$ became stronger,
    and we need to show that $K'+2\le K_n$. However, since $K'$ and $K_n$ are of the same parity,
    this is a consequence of $K'<K_n$.
    
    \end{proof}
    The claim proves that lines 7-11 are correct.
    Because $K_{\ell}$ and $K_{\ell+1}$ satisfied $K_{\ell}<K_{\ell+1}$, that is, $K_{\ell}+2\le K_{\ell+1}$,
    before these lines were executed
    and each $K_i$ changed by at most 1, they now satisfy $K_{\ell}\le K_{\ell+1}$,
    so the monotonicity is also preserved.
    Also, observe that after operations in lines 7-11 all constant terms have the same parity.

    Finally, observe that each iteration decreases $\sum_i (1+D_i)$ and thus the procedure terminates.
    
    \end{proof}

    \SetKwFunction{FSolveTwo}{Solve2}

    Now we transform the \FSolve function into a more efficient version which we call \FSolveTwo.
    The latter simulates selected series of iterations of the while-loop of the former
    in single steps of its while-loop.

\SetKwFunction{FSolveTwo}{Solve2}
 \begin{function}
 \caption{Solve2($K',K_1,\ldots,K_n, D_1,\ldots,D_n$)}
 \lIf{$n=1$}{\Return $\max(K',K_1)$}
 \While{$K'<K_n$}{
 $\ell := \max\{i : K_i = K_1\}$\;
 \For{$i=1$ \KwSty{to} $\ell$}{
 	\If{$D_i =0$}{\Return{\FSolveTwo$(K',K_1,\ldots,K_{i-1},K_{i+1},\ldots,K_n, D_1,\ldots,D_{i-1},D_{i+1},\ldots,D_n)$}\;}
 }
 \lIf{$\ell<n$ \KwSty{and} $K_n-K' < 2\ell$}{\Return{$\frac{1}{2}(K'+K_n)$}}
 \lIf{$\ell=n$ \KwSty{and} $K_n-K'< 2(n-1)$}{\Return{$1+\frac{1}{2}(K'+K_n)$}}
	$K' := K'+\ell$\;
	\For{$i=1$ \KwSty{to} $\ell$}{
		$K_i := K_i + 2-\ell$\;
		$D_i := D_i-1$\;
	}
	\lFor{$i=\ell+1$ \KwSty{to} $n$}{
		$K_i := K_i-\ell$}
}
\Return{$K'$}\;
\end{function}

\begin{lemma}
Functions \FSolveTwo and \FSolve are equivalent for all inputs corresponding to easy \pmilp{s} in the canonical form with $K_1\le \ldots \le K_n$.
\end{lemma}
\begin{proof}
Lines 1-6 of  \FSolve and \FSolveTwo coincide, so we can restrict to inputs for which
both functions reach line 7.
Recall from the proof of Lemma~\ref{lem:solve} that this implies $n \ge 2$, $K'<K_n$ and $D_1,\ldots,D_\ell\ge 1$.
We claim that lines 7-13 simulate the first $\ell$ iterations of the while-loop of \FSolve
(or all $<\ell$ iterations if the loop execution is interrupted earlier).
Note that the value of $\ell$ decrements by 1 in each of these iterations excluding the last one.
Also, in each of these iterations the condition in line 5 is false, since in the previous iteration
the respective $D_i$ variables did not change.
Thus, after an iteration with $\ell>1$, the execution
of the following $\ell-1$ iterations of the while-loop may only be interrupted
due to the $K'<K_n$ condition (line 2).

Note that if $\ell<n$ then the difference $K_n-K'$ decreases by exactly 2 per iteration,
as $K_n$ is decremented while $K'$ is incremented.
Thus, starting with $\ell<n$ the loop fails to make $\ell$ steps only if $K_n-K' < 2\ell$
and it happens after $\frac{1}{2}(K_n-K')$ iterations (due to equal parity of constant terms this is an integer).
Then, the value $K'$, incremented at each iteration, is $\frac{1}{2}(K'+K_n)$ in terms of the original values.

On the other hand, if $\ell=n$, then in the first iteration $K_n-K'$ does not change, since $K'$ and $K_n$ are both
incremented. The while-loop execution is interrupted before $\ell$ iterations only if $K_n-K'< 2(n-1)$,
and this happens after $1+\frac{1}{2}(K_n-K')$ iterations, when $K'$, incremented at each iteration, is $1+\frac{1}{2}(K'+K_n)$ in terms of the original values.

Consequently, lines 7 and 8 are correct, and if \FSolveTwo proceeds to line 9, we are guaranteed that \FSolve would
execute at least $\ell$ iterations, with value $\ell$ decremented by 1 in each iteration.
It is easy to see that lines 9-13 of \FSolveTwo simulate execution of lines 7-11 of \FSolve in these $\ell$ iterations,
which corresponds to simultaneously replacing $y_1,\ldots,y_{\ell}$ with $y_{1}+1,\ldots,y_{\ell}+1$.
\end{proof}

\SetKwFunction{FSolveThree}{Solve3}

Finally we perform the final improvement of the algorithm for solving our special ILP,
which yields an efficient function \FSolveThree.
Again what actually happens is that series of iterations of the while-loop from \FSolveTwo
are now performed at once.

 \begin{function}
 \caption{Solve3($K',K_1,\ldots,K_n, D_1,\ldots,D_n$)}
 \lIf{$n=1$}{\Return $\max(K',K_1)$}
 \While{$K'<K_n$}{
 $\ell := \max\{i : K_i = K_1\}$\;
 \For{$i=1$ \KwSty{to} $\ell$}{
 	\If{$D_i =0$}{\Return{\FSolveThree$(K',K_1,\ldots,K_{i-1},K_{i+1},\ldots,K_n, D_1,\ldots,D_{i-1},D_{i+1},\ldots,D_n)$}\;}
 }
\If{$\ell<n$}{
	$\Delta := \min(D_1,\ldots,D_{\ell}, \frac{1}{2}(K_{\ell+1}-K_{\ell}), \floor{\frac{1}{2\ell}(K_n-K')})$\;
	\lIf{$\Delta = 0$}{\Return{$\frac{1}{2}(K'+K_n)$}}
}\Else{
	$\Delta := \min(D_1,\ldots,D_n, \lfloor{\frac{1}{2(n-1)}(K_n-K')\rfloor})$\;
	\lIf{$\Delta = 0$}{\Return{$1+\frac{1}{2}(K'+K_n)$}}
}
$K' := K'+\Delta\cdot\ell$\;
\For{$i=1$ \KwSty{to} $\ell$}{
		$K_i := K_i + \Delta(2-\ell)$\;
		$D_i := D_i-\Delta$\;
	}
	\lFor{$i=\ell+1$ \KwSty{to} $n$}{
		$K_i := K_i-\Delta\cdot\ell$}
}
\Return{$K'$}\;
\end{function}
\begin{lemma}\label{lem:eq23}
Functions \FSolveThree and \FSolveTwo are equivalent for all inputs corresponding to easy \pmilp{s} in the canonical form with $K_1\le \ldots \le K_n$.
Moreover, \FSolveThree runs in $\Oh(n^2)$ time (including recursive calls).
\end{lemma}
\begin{proof}
Lines 1-6 of \FSolveTwo and \FSolveThree coincide, so we can restrict to inputs for which
both functions reach line 7. This in particular means that $D_1,\ldots,D_\ell \ge 1$.
Also, the definition of $\ell$ and the equal parity of constant terms mean that $K_{\ell+1}\ge K_{\ell}+2$.
Thus, if $\Delta=0$ in line 9, then $K_n-K'<2\ell$, and similarly in line 12,
$\Delta=0$ only if $K_n-K'<2(n-1)$. These are exactly
conditions that \FSolveTwo checks in lines 7 and 8, and consequently \FSolveThree simulates the execution of \FSolveTwo if
one of these conditions is satisfied.

Consequently, in the following we may assume that $\Delta>0$.
Observe that during the execution of the while-loop in \FSolveTwo,
the value $\ell$ may only increase. Indeed, the operations in lines 9-13 preserve the equality $K_1=\ldots=K_{\ell}$
and decrease by 2 the difference $K_{\ell+1}-K_{\ell}$, which leads to an increase of $\ell$ after $\frac{1}{2}(K_{\ell+1}-K_{\ell})$
steps.
Also, in a single iteration, the difference $K_n-K'$ decreases by either $2\ell$ if $\ell<n$ or $2(n-1)$ if $\ell=n$.
Consequently, unless the execution is interrupted, each value included in the minima in lines 8 and 11
decreases by exactly 1.
Moreover, the execution is interrupted only if one of these values reaches 0,
so we are guaranteed that $\Delta$ iterations will not be interrupted and can simulate them in bulk.
This is what happens in lines 13-17 of \FSolveThree, which correspond to lines 9-13 of \FSolveTwo.
Thus, \FSolveTwo and \FSolveThree are indeed equivalent. It remains to justify the quadratic
running time of the latter.

\medskip
Observe that each iteration of the while-loop of \FSolveThree can be executed in $\Oh(n)$ time.
Thus, it suffices to show that, including the recursive calls, there are $\Oh(n)$ iterations in total.
We show that $2n-\ell$ decreases not less frequently than every second iteration.
First, note that the recursive call decrements both  $n$ and $\ell$ by one, and thus it
also decrements $2n-\ell$. Therefore it suffices to consider an iteration
which neither calls \FSolveThree recursively nor terminates the whole procedure.
Then lines 13-17 are executed. If $\Delta$ was chosen so that $\Delta=D_i$,
the next iteration performs a recursive call, and thus $2n-\ell$ decreases after this iteration.
On the other hand, if $\Delta = \frac{1}{2}(K_{\ell+1}-K_{\ell})$, then after executing lines 13-17
it holds that $K_{\ell}=K_{\ell+1}$ and thus in the next iteration $\ell$ is increases, while $n$ remains unchanged,
so $2n-\ell$ decreases. Finally, if the previous conditions do not hold, then $\Delta =  \floor{\frac{K_n-K'}{2\min(\ell,n-1)}}$ and in the next iteration
line 7 is reached again, but then in line 8 or 11 $\Delta$ is set to 0, and the procedure is terminated
in line 9 or 12. This shows that it could not happen that three consecutive iterations
(possibly from different recursive calls) do not decrease $2n-\ell$.
This completes the proof that \FSolveThree executes $\Oh(n)$ iterations in total.
\end{proof}

Lemmas~\ref{lem:solve}-\ref{lem:eq23} together show that \FSolveThree
correctly solves any \pmilp in the canonical form in $\Oh(n^2)$ time.
Moreover, note that the algorithm is simple and has a small constant hidden in the $\Oh$
notation. Also, with not much effort, the optimal solution $(y_1,\ldots,y_n)$ can be recovered.
Combined with a reduction of Lemma~\ref{lem:red}, we obtain the following result.
\begin{corollary}
Any $n$-variate easy \pmilp can be solved in $\Oh(n^2)$ time.
\end{corollary}

\medskip
\noindent
The natural open question is whether our approach
extends to larger values of $k$, e.g.\ to $k=6$. It seems
that while combinatorial properties could still be used
to dramatically reduce the number of interval systems required to be processed,
the resulting {\pmilp}s are not easy and have more variables.
Consider, for example, a special case where for each column
there are just 2 different values, each occurring 3 times. Then, there is a single
interval sequence consistent with each $\x$ satisfying Observation~\ref{obs:bet},
it also corresponds to taking a median in each column.
Nevertheless, after simplifying the ILP
there are still $\frac{1}{2}{6 \choose 3}=10$ variables, and the \pmilp is not an easy \pmilp.

\end{document}